\documentclass[xcolor=svgnames,12pt]{amsart}

\newtheorem{theorem}{Theorem}[section]
\newtheorem{lemma}[theorem]{Lemma}
\newtheorem{corollary}[theorem]{Corollary}
\newtheorem{proposition}[theorem]{Proposition}
\newtheorem{assumption}[theorem]{Assumption}

\theoremstyle{definition}
\newtheorem{remark}[theorem]{Remark}
\newtheorem{definition}[theorem]{Definition}

\numberwithin{equation}{section}

\renewcommand{\(}{\left(}
\renewcommand{\)}{\right)}
\renewcommand{\[}{\left[}
\renewcommand{\]}{\right]}
\renewcommand{\tilde}{\widetilde}
\newcommand{\1}{\mathbf{1}}

\newcommand{\e}{{\mathbb{E}}}
\newcommand{\F}{{\mathcal{F}}}
\newcommand{\q}{\mathbb{Q}}
\newcommand{\X}{\mathbb{X}}

\newcommand{\pnl}{\textup{P\&L}}
\newcommand{\R}{{\mathbb{R}}}
\usepackage{textcomp} % for \textlangle and \textrangle macros
\newcommand{\qdist}[1]{\ifmmode\langle#1\rangle\else\textlangle#1\textrangle\fi}
\newcommand{\A}{{\mathbb{A}}}

\newcommand{\dd}{{\textup{d}}}
\newcommand{\blue}{\color{blue}}
\newcommand{\red}{\color{red}}

\newcommand{\spicy}{\color{spicy}}

\usepackage{comment}
\usepackage{amsmath}
\usepackage{amsfonts}
\usepackage{amssymb}
\usepackage{nicefrac}
\usepackage{framed}
\usepackage{color}
\usepackage{pstricks}
\usepackage{ipa}
\usepackage{textcomp}
\usepackage{url}
\usepackage{stmaryrd}
\usepackage{mathabx}
\usepackage{xcolor}
\usepackage{color}
\usepackage{graphicx}
\usepackage{subfig}
\usepackage[left=2.2cm,top=2.8cm,right=2.2cm,bottom=2.8cm,head=3cm,headsep=1.3cm, foot=3cm]{geometry}
\usepackage[bottom,flushmargin,hang,multiple]{footmisc}
\usepackage{setspace}
\usepackage{xr-hyper}
\usepackage{hyperref}
\usepackage{commath}
\usepackage{bm}
\usepackage{enumitem}
\usepackage{booktabs}

\newcommand{\ba}{\begin{eqnarray}}
\newcommand{\bs}{\textup{BS}}
\newcommand{\ea}{\end{eqnarray}}
\newcommand{\bas}{\begin{eqnarray*}}
\newcommand{\eas}{\end{eqnarray*}}

\definecolor{spicy}{HTML}{B51D0A}

\definecolor{darkblue}{rgb}{0.1,0.1,0.7}
\definecolor{darkred}{rgb}{0.9,0.1,0.1}

\hypersetup{colorlinks, allcolors= darkblue}

\makeatletter

\newcommand{\Rmnum}[1]{\expandafter\@slowromancap\romannumeral #1@}
\makeatother

\begin{document}
\title[VAs of Variable Annuities]{Out-of-model adjustments of Variable Annuities}

\author[Zhiyi Shen]{
Zhiyi Shen \vspace{0.1cm} \\ 
% Fixed Income Division  \vspace{0.1cm} \\ 
Morgan Stanley \vspace{0.4cm}\\
This draft: \today \vspace{0.6cm}
}

% \address{Berkeley}
% \email{\href{zhiyi.shen@uwaterloo.ca}{zhiyi.shen@uwaterloo.ca}}

% \thanks{\textit{Key Words:} Valuation Adjustment, Variable Annuities, Model Risk.} 
\thanks{\textit{Email address: \href{zhiyi.shen@uwaterloo.ca}{zhiyi.shen@uwaterloo.ca} }}

\thanks{\vspace{0.15cm} }
 
% \thanks{\textit{2010 Mathematics Subject Classification:} Template. \vspace{0.15cm}}

\thanks{The views expressed in this article are the author's own 
and do not represent the opinions of any firm or institution.}

%====================================================================================
%====================================================================================
%====================================================================================

\begin{abstract}
This paper studies the model risk of the Black-Scholes (BS) model in
pricing and risk-managing variable annuities motivated by its wide usage in the insurance industry.
Specifically, we derive a model-free decomposition of the no-arbitrage price 
of the variable annuity into the BS model price in conjunction with three out-of-model adjustment terms.
This sheds light on all risk drivers behind the product, that is, spot price, realized volatility,
future smile, and sub-optimal withdrawal.
We further investigate the efficacy of the BS-based hedging strategy
given the market diverges from the model assumptions.
We disclose that the spot price risk can always be eliminated by the strategy
and the hedger's cumulative P\&L exhibits gradual slippage and instantaneous leakage.
We finally show that the pricing, risk and hedging models can be separated from each other 
in managing the risks of variable annuities.\\
\vspace{1ex}
\end{abstract}

\maketitle
                                    
%====================================================================================
%====================================================================================
%====================================================================================

\newpage

\section{Introduction}
Variable annuities are long-term, equity-linked, and tax-deferred structure products 
issued by insurance companies targeting retail customers.
The size of the U.S. variable annuities market is remarkable.
By the end of 2021, variable annuity net assets in the U.S. climbed to 2,130 billion dollars \cite{IRI2021},
around two-thirds of the notional amount outstanding
of the entire U.S. OTC equity derivatives market with 3,567 billion dollars \cite{BIS2021}.

The pricing of variable annuities is a stochastic control problem 
without analytical solution in general \cite{Azimzadeh2015,Huang2016,Huang2017,Shen2020}
and thus is fairly cumbersome even under the classical Black-Scholes (BS) model.
Despite extensive studies on various numerical methods for pricing variable annuities,
little understanding has been delivered in the literature regarding the impact of model misspecification 
on the pricing and hedging of variable annuities.
This paper aims to bridge this gap by exploring tentative answers to the following questions
raised from different aspects.
\begin{itemize}
    \item From the perspective of pricing, the pricer needs to have a thorough understanding of how to determine
    the volatility parameter in the BS model to conservatively price the product 
    to compensate for the potential model risk. 

    \item From the hedger's perspective, a natural question: is it still viable
    to conduct BS-Delta-hedging given the market diverges
    from the model assumptions?
    Does doing so reduce the risk or even escalate the problem?

    \item From the viewpoint of risk management, it is important to pinpoint all
    risk drivers behind the variable annuity. Then the insurer can decide which 
    risk to hedge, which risk to outsource, and which risk to take. 
\end{itemize}

In response to the questions above, this paper makes several contributions to the literature.
As the primary contribution, we decompose the model-free no-arbitrage price of the variable annuity
into the BS model price in together with
(i) valuation adjustment for future realized volatility, 
(ii) valuation adjustment for future implied volatility smile,
and (iii) valuation adjustment for sub-optimal withdrawal risk;
see Theorems \ref{thm:risk_decompose} and \ref{thm:va_multi_period} in Section \ref{sec:valuation_adjustment}.
We further show that the BS model enables the insurer to speculate
the volatility risks by marking up/down the volatility parameter.

As the second contribution, we investigate the efficacy of BS-based delta-hedging in the presence of model risk. 
We find that the risk caused by the underlying asset's price change
can always be eliminated by such a classical hedging strategy
\textit{regardless of whether the market behaves in accordance with the BS model's assumptions or not}.
Furthermore, there is even a chance that the hedger can benefit from taking the model risk; see Proposition \ref{prop:carry_pnl}.
This justifies the use of the BS model as a hedging tool to some extent.
However, such a hedging strategy does not come without any downsides.
We disclose that the hedger's cumulative P\&L exhibits gradual slippage throughout the contract's lifetime
and instantaneous leakage across each withdrawal date; see Remark \ref{rem:pnl_leakage} in the sequel.

As the third contribution, we show that it is viable to separate the risk model from 
the pricing model.
Specifically, on the one hand, the insurer may solely use the BS model as an extractor
for the spot price risk with the residual part further decoupled into three extra factors,
i.e., realized volatility, future implied volatility and sub-optimal withdrawal.
Such a way of risk attribution is \textit{exhaustive} and can be constantly monitored;
see Remark \ref{rem:risk_attribution}.
On the other hand, the insurer may use a different pricing model to charge the premium
or estimate the hedging costs for the four risk factors.

Finally, we would like to stress that the notion of out-of-model-adjustment is not new
and has been well understood for European options 
\cite{Andersen2010,Bergomi2015,Carr2001,Gatheral2011} and cliquets \cite{Bergomi2015}.
However, it is a fairly challenging task to find a systematic way to decompose 
the model-free price of any exotic derivative product
into a given model price plus adjustments that have financial meanings 
and can shed light on different aspects of model risk.
For variable annuities, to the best of the author, 
this paper is the first attempt to pursue this route.
The existing results mentioned above cannot be carried over here
because the variable annuity contains some unique risk features, and in particular,
it gives rise to the valuation adjustment for sub-optimal withdrawal risk
that is unseen in other products.
This distinguished risk profile stems from
the stochastic control problem involved in variable annuities.

The remainder of this article is structured as follows.
Section \ref{sec:variable_annuity} gives a brief recap on the pricing of the variable annuity
as a stochastic control problem.
Section \ref{sec:valuation_adjustment} presents the main result of the paper:
the out-of-model adjustment formula.
Section \ref{sec:num_studies} gives numerical studies
and finally Section \ref{sec:conclusion} concludes the paper.

\vspace{2ex}
\section{Variable Annuities}
\label{sec:variable_annuity}
\subsection{Notations}
\begin{itemize}
 \vspace{1ex}   \item Consider a set of equally-spaced withdrawal dates  $\mathcal{T}:=\{t_i\}_{i=1}^{N}$ with $\delta:= t_{i+1} - t_i$.
 
 \vspace{1ex}\item Let $X(t)$ be the time-$t$ value of the state variable associated with the contract (investment account, benefit base, etc) valued in $\X$. To ease the notations, in the sequel, we denote shorthand $X_n:=X(t_n)$ and $X_{n^+}:=X\(t_n^{+}\)$.
 \begin{remark}
 $X(t)$ is not necessarily a scalar process. For the clarity of presentation, this paper concentrates on the one-dimensional case, which however is not essential to our argument for proving the main results.
\end{remark}

\vspace{1ex}\item
Denote $K: \X\times \A\rightarrow \X$ as the transition map of the state variable across an event date,
with $\A$ being the feasible set of the withdrawal policy.
That is,
\bas
X_{n^+} = K(X_n, a),\ \ a \in A_n(X_n) 
\eas
where $a$ is the policyholder's withdrawal amount at $t_n$
and $A_n(X_n) \subset \A$ is some state-dependent constraint.
% Let the set-valued function $A_n: \mathbb{X} \rightrightarrows \mathbb{A}$ be the state-dependent constraint of the withdrawal policy at $t_n$.

\vspace{1ex}\item
Between two withdrawal dates, the state variable evolves according to
\bas
X(t) = F_n\(X_{n^{+}}, \varepsilon(t)\),\ \ t \in {\blue(t_n}, t_{n+1}],
\eas
where $\varepsilon(t) \in \F_t$ is some random driver 
(e.g. the cumulative return of the underlying asset over $[t_n, t]$).

\vspace{1ex}\item 
Let $g: \mathbb{X} \rightarrow \R_+$ and $f_n: \X \times \A \rightarrow \R_+$ be the terminal and intermediate payoff functions, respectively.

\vspace{1ex}\item
Let $r$ be the risk-free rate and denote by $\q$ the risk-neutral measure.
Let $\F_t$ be the information up to time $t$.
\end{itemize}

\subsection{Bellman Equation -- Model-free Case}
The pricing of variable annuities is typically formulated as a discrete-time stochastic control problem 
and accordingly, the value function at withdrawal time is recursively given by the following Bellman equation:
\ba
\label{eq:Bellman_eq}
\begin{cases}
V(t_N) &= g(X_N),\\
V(t_n) &= \sup\limits_{a \in A_{n}\(X_n\)} \[f_{n}(X_n, a) + C_{n}\(K(X_n, a)\)\],
\ \ 1\leq n \leq N-1,
\end{cases}
\ea
where 
\ba
\label{eq:post_withdrawal_value}
C_{n}(x) = \e_{n,x}^{\q} \[e^{-r\delta}V(t_{n+1})\],
\ea
with
$
\e_{n,x}^{\q}[\cdot]:=\e^{\q}\[\cdot|\mathcal{G}_{n}^{x}\]
$
and $\mathcal{G}_{n}^{x}:=\sigma\(\{X_{n^{+}}=x\}\bigcup \mathcal{F}_{t_n}\)$;
see e.g. \cite{Azimzadeh2015,Huang2016,Huang2017,Shen2020}. 
$C_n(x)$ can be thought of as the value of the contract right after the policyholder's withdrawal at $t_n$ 
with the post-withdrawal state $X_{n^+}=x$.

Between two event dates, the value function satisfies
\ba
\label{eq:martingale_eq}
V(t) = \e^{\q} \[\left. e^{-r(t_{n+1}-t)}V(t_{n+1})\right|\mathcal{F}_t\],\ \ t\in{\blue(t_n}, t_{n+1}].
\ea
As a remark, with the convention that the contract inception time $t_0$ is not a withdrawal date, the above equation holds for $t\in{\blue [t_0}, t_{1}]$ when $n=0$.
\subsection{Bellman Equation -- Black-Scholes Case}
In a BS world, the price function of the variable annuity is recursively given by
\ba
\label{eq:Bellman_eq_BS}
\begin{cases}
V_{\bs}(t_N, x) &= g(x),\\
V_{\bs}(t_n, x, v) &= \sup\limits_{a \in A_{n}\(x\)} \[f_{n}(x, a) + C_{n}^{\bs}\(K(x, a),v\)\],
\ \ 1\leq n \leq N-1,
\end{cases}
\ea
where 
\ba
\label{eq:post_withdrawal_BS}
C_{n}^{\bs}(x,v) = V_{\bs}\({\blue t_{n}^{+}}, x,v\),
\ea
and in between two withdrawal dates $V_{\bs}$ solves the following BS-type PDE:
\ba
\label{eq:BS_eq}
\begin{cases}
V_{\bs}|_{t=t_{n+1}} = V_{\bs}({\blue t_{n+1}}, x, v),&\\
\partial_t V_{\bs} + rx \partial_x V_{\bs} + \frac{v}{2}x^2\partial^2_{xx} V_{\bs}- rV_{\bs} = 0,\ \ {\blue t\in(t_n, t_{n+1})},\ \ 0\leq n \leq N-1.&
\end{cases}
\ea
with $v$ denoting the BS variance parameter;
cf. \cite{Azimzadeh2015}.
\begin{remark}[Value Function vs. Value Process]
It is worth stressing the difference between $V_{\bs}(t, x, v)$ and $V(t)$:
the former is a value function and while the latter is a stochastic process
which is not necessarily a function of the state variable $X(t)$ due to the existence of other (latent) risk factors generating the filtration $\F=\(\F_t\)_{t\in[0, t_N]}$.
Similarly, given a realized value of $X_{n^+}=x$, $C_n(x)$ is still a random variable 
(parameterized by $x$) while $C_{n}^{\bs}(x, v)$ is deterministic.
\end{remark}
% \begin{remark}
% Slightly abusing the notation, we will denote the BS value function and post-withdrawal value by $V_{\bs}(t, x, v)$ and $C_n^{\bs}(x, v)$ respectively to stress their dependency on $v$ when necessary.
% \end{remark}
The primary thrust of the article is to characterize the discrepancy between $V_{\bs}(t,x,v)$ and $V(t,x)$.

\vspace{2ex}
\section{Out-of-Model Adjustments}
\label{sec:valuation_adjustment}
This section is devoted to presenting the main results of the article.
To get some flavours of the problem, 
we start with a two-period case in the subsequent Section \ref{sec:two_period_case}.
\subsection{Valuation Adjustments}
\label{sec:two_period_case}
\begin{theorem}[Out-of-Model Adjustments]
\label{thm:risk_decompose}
Let $N=2$ and $\xi_1 \in \F_{t_1}$ denote the BS implied variance of an European option with payoff $g(X(t_2))$ as of time $t_1$ (see also Definition \ref{def:imp_var} in Appendix \ref{app:proof_multi_period}).
Suppose the assumptions in Appendix \ref{app:assum} hold.
Then 
\ba
\label{eq:va_eq}
V(t) = V_{\bs}(t, X(t), v)+ V_{1}(t) + V_{2}(t), \ \ t \in [t_0, t_1],
\ea
where
\ba
\label{eq:gamma_va}
V_{1}(t) &=&\e^{\q}\[\left.\int_t^{t_1} \frac{1}{2}e^{-r(u - t)} 
\partial_{x}^2 V_{\bs}\(u, X(u), v\)\[\dd\[X\](u)-vX^2(u)\dd u\]\right|\F_t\],\\
\label{eq:smile_va}
V_{2}(t) &=& e^{-r(t_1 - t)}\e^{\q}\[\left.V_{\bs}\(t_1, X_1, {\blue\xi_1}\) - V_{\bs}\(t_1, X_1, {\blue v}\)\right|\F_t\],
\ea
and $\[X\](t)$ denotes the quadratic variation process of $X(t)$.
\end{theorem}
\begin{remark}[Realized \& Implied Volatility Risks]
Theorem \ref{thm:risk_decompose} discloses two different risks that are not priced in the BS model.
\begin{itemize}
\item {\it Realized Volatility Risk}\quad
The first valuation adjustment (VA) term \eqref{eq:gamma_va} is contributed by the discrepancy between
the BS variance parameter $v$ and the \textit{realized variance} $\(\dd X(u)/X(u)\)^2$ of the underlying.
Should the underlying price behave in accordance with the BS model, this VA term would vanish.

\vspace{1ex}\item {\it Future Smile Risk}\quad
The second VA term \eqref{eq:smile_va} is caused by the difference between BS variance $v$
and the \textit{implied variance} of the European option with payoff $g(X(t_2))$ 
at a \textit{future} time point $t_1$.
The latter is determined by the entire volatility smile as of $t_1$ and for this reason,
we refer to this risk source as the \textit{future-smile risk}.
In an ideal BS world, all European options have the same implied variance and thus this VA term disappears.
\end{itemize}
To summarize, we have
\bas
\textup{Fair price} = \textup{BS price} + \textup{VA for future realized volatility risk} + \textup{VA for future smile risk}.
\eas
The two VA terms vanish should the model's assumptions be satisfied.
\end{remark}

\begin{remark}
\label{rem:breakeven}
Evaluating the two VA terms calls for the knowledge of the true model, which is a formidable task.
Nonetheless, gauging their signs is relatively easier:
\begin{enumerate}
\item The VA for the future realized volatility risk is positive (resp. negative) if the the BS model variance $v$ underestimates (resp. overestimates) the {\it realized} variance $\(\dd X(u)/X(u)\)^2$ over $[t, t_1]$;

\vspace{1ex}\item The VA for the future smile risk is positive (resp. negative) if the the BS model variance $v$ underestimates (resp. overestimates) the {\it implied} variance $\xi_1$.
\end{enumerate}
The critical insight from the above is that the 
insurer has the freedom of under-pricing or over-pricing the variable annuity by marking up/down the BS variance parameter $v$.
In other words, {\it intentionally or not, should the insure decide to choose the BS model to price the variable annuity, 
she essentially speculates the realized volatility and future smile risks.}
\end{remark}

\subsection{P\&L Slippage and Leakage}
The previous discussion reveals that the BS model fails to price in Gamma and future smile risks.
Next, we study the impact of the use of BS as a hedging tool on the insurer's cumulative P\&L.

Consider the following situation: the BS model is not in line with the market's dynamics
but the insurer pretends it to be and delta-hedges her exposure to the variable annuity.
Specifically, the hedging strategy is given as follows.
\begin{itemize}
\vspace{1ex}\item[] {\bf (H1)} At time $t_0$, the insurer sells a variable annuity contract and chooses to value it by the BS model with a BS variance parameter $v$. Accordingly, the premium received by the issuer is $V_{\bs}(t_0)$ which is used to finance her hedging strategy.

\vspace{1ex}\item[] {\bf (H2)} Over the time horizon $[t_0, t_1)$, the insurer continuously delta-hedges her exposure with two commitments:
\begin{enumerate}
\vspace{1ex} \item[] {\bf (C1)} the insurer freezes the variance parameter $v$
because the implied volatility of the variable annuity is not quoted in the market;

\vspace{1ex} \item[] {\bf (C2)} for the same reason as the above, the insurer always marks her position to the model price $V_{\bs}(t, X(t), v)$.
\end{enumerate}
\vspace{1ex} \item[] {\bf (H3)} At time $t_1$, the variable annuity degenerates into an European option whose implied variance $\xi_1$ 
can be observed from the market quotes for vanilla options\footnote{
Given a market for vanilla call options with all strikes, 
one can pin down the BS implied variance/volatility for any European option with convex payoff \cite{Bergomi2015}.}. 
Accordingly, the insurer can mark the value of the contract be $V_{\bs}(t_1, X_1, \xi_1)$.
\end{itemize}

Now the problem of interest is to understand how the insurer's mark for P\&L varies as time progresses from $t_0$ to $t_1$. The following proposition sheds light on this.
\begin{proposition}
\label{prop:carry_pnl}
Suppose the assumptions in Appendix \ref{app:assum} hold.
Let $\Pi(u)$ be the value of the insurer's hedges as of time $u$.
The cumulative profit and loss marked by the insurer is given by
\ba
\label{eq:carry_pnl}
\textup{P\&L}(t) = \Pi(t) - V_{\bs}\(t, X(t), v\)
= \textup{P\&L}_{\Gamma}(t),
 \ \ t \in [t_0, {\blue t_1)},
\ea
where
\ba
\label{eq:Gamma_pnl}
\textup{P\&L}_{\Gamma}(t) = \int_{t_0}^{t} \frac{e^{r(t-u)}}{2}\partial_{x}^2V_{\bs}\[vX^2(u)\dd u - \dd\[X\](u)\],
\ \ t \in [t_0, t_1].
\ea

Furthermore, the cumulative profit and loss realized by the insurer as of $t_1$ is given by
\ba
\label{eq:mtm_pnl}
\textup{P\&L}({\blue t_1}) &=& \Pi(t_1) - V_{\bs}\(t_1, X_1, \xi_1\)
= \underbrace{\textup{P\&L}_{\Gamma}(t_1)}_{\textup{P\&L slippage}}
+ \underbrace{\blue \textup{P\&L}_{L}(t_1)}_{\textup{P\&L leakage}},
\ea
where
\ba
\label{eq:pnl_leakage}
\textup{P\&L}_{L}(t_1) = V_{\bs}\(t_1, X_1, v\) - V_{\bs}\(t_1, X_1, \xi_1\).
\ea
\end{proposition}
% \begin{enumerate}
\begin{remark}[P\&L Slippage and Leakage]
\label{rem:pnl_leakage}
Proposition \ref{prop:carry_pnl} discloses two P\&L impacts with different natures brought by the use of the BS-based hedging strategy.
\begin{itemize}
\vspace{1ex}\item {\it Slippage}\quad Before time $t_1$, as time $t$ progresses, 
the insurer can gradually feel the mis-hedging of the BS model
because she can observe exposure to Gamma risk (Eq. \eqref{eq:Gamma_pnl}) in her marked P\&L (Eq. \eqref{eq:carry_pnl})
which shouldn't have come into place had the market behaved in accordance with the BS model.
We refer to the term \eqref{eq:Gamma_pnl} as the P\&L \textit{slippage}
to stress this incremental bleeding.

\vspace{1ex}\item {\it Leakage}\quad Recall from {\bf (H3)} that at time $t_1$ the insurer can observe the fair market price of the variable annuity and her final P\&L reads \eqref{eq:carry_pnl}.
By comparing Eqs. \eqref{eq:carry_pnl} and \eqref{eq:carry_pnl}, the insurer sees a sudden mark up/down for her position across $t_1$
caused by the second term in Eq. \eqref{eq:mtm_pnl}. To stress this discontinuity in $\textup{P\&L(t)}$ across $t_1$, we refer to the term \eqref{eq:pnl_leakage} as the P\&L \textit{leakage}.
\end{itemize}
\end{remark}
\begin{remark}
In the classical BS paradigm, 
the model's assumptions are supposed to be fulfilled by the market
and thus the Gamma risk enters into the hedger's P\&L only when the hedging frequency is discrete. 
Proposition \ref{prop:carry_pnl} reveals that the Gamma exposure is generally inevitable 
even if we assume continuous hedging due to the presence of model risk.
In a more realistic situation, we shall expect extra P\&L and valuation adjustment terms taking into account
the impact of discrete re-balancing.
\end{remark}

The conclusion of Proposition \ref{prop:carry_pnl} is not surprising:
we have shown in Theorem \ref{thm:risk_decompose} that the fair price contains two VA terms on top of the BS model price,
which implies that if the insurer only charges the BS model price as the premium 
it might overestimate/underestimate the entire hedging cost as reflected by the P\&L slippage and leakage terms illustrated in the above.
Nevertheless, the insurer does not necessarily lose/make money depending on the relative order between the parameter $v$
and implied/realized variance; see Remark \ref{rem:breakeven}.

\subsection{Multi-period Case}
\label{sec:multi_period}
The following theorem generalizes Theorem \ref{thm:risk_decompose} to the case with multiple withdrawal dates.
\begin{theorem}
\label{thm:va_multi_period}
Suppose the assumptions in Appendix \ref{app:assum} hold.
\ba
\label{eq:va_eq_multi_period}
V(t) = V_{\bs}(t, X(t), v)+
V_{1}(t) + V_{2}(t) + V_3(t), \ \ t \in [t_{n-1}, t_n],
\ \ 1\leq n \leq N-1,
\ea
where $V_i(t), i=1,2,3,$ are given in Eqs. \eqref{eq:va_gamma_process}--\eqref{eq:va_suboptimal_process} of Appendix \ref{app:proof_multi_period}.
\end{theorem}
\begin{remark}[Sub-optimal Withdrawal Risk]
\label{rem:suboptimal_withdrawal_risk}
With respect to the two-period case (see Eq. \eqref{eq:va_eq}),
Eq. \eqref{eq:va_eq_multi_period} discloses one extra valuation adjustment term $V_{3}(t)$
which stems from the fact that the optimal withdrawal strategy associated with
the ``true'' model is only suboptimal under the BS model; see Eq. \eqref{eq:va_suboptimal_process}
for a concise definition of the sub-optimal withdrawal risk.

Generally speaking,  
it is not surprising that the optimal withdrawal strategy of any given model 
is not in line with the one observed from the market due to the existence of model risk.
Such a disagreement has been attributed to the irrationality of the policyholder
or tax considerations in the literature; see
\cite{Chen2008,Moenig2016,Knoller2016} and the references therein.
This article provides an alternative explanation, that is, the model mis-specification risk.
\end{remark}

The following corollary directly follows from the above theorem.
\begin{corollary}
Under the conditions of Theorem \ref{thm:va_multi_period}, we have
\ba
\label{eq:pnl_attribution}
\nonumber
\dd V(t) &=& 
\underbrace{\partial_t V_{\bs}(t, X(t), v) \dd t}_{\textup{time decay}} 
+ \underbrace{\partial_x V_{\bs}(t, X(t), v)\dd X(t)}_{\textup{Delta effect}}\\ 
&&+ \underbrace{\frac{1}{2}\partial_{x}^2 V_{\bs}(t, X(t), v)\dd [X](t)}_{\textup{realized vol effect}} 
+ \underbrace{\dd V_{2}(t)}_{\textup{future smile effect}}
+ \underbrace{\dd V_{3}(t),}_{\textup{withdrawal effect}}
\ea    
where the expressions of $V_2(t)$ and $V_3(t)$ are relegated to Definition \ref{def:va_adj} 
of Appendix \ref{app:proof_multi_period} for the clarity of presentation.
\end{corollary}

\begin{remark}[Exhaustive Risk Attribution]
\label{rem:risk_attribution}
The primary message delivered by the above corollary is that
the P\&L of the variable annuity can be \textit{exhaustively} attributed to four drivers
(i) spot price\footnote{To be precise, we refer
to P\&L that is solely caused by the first-order-change of the underlying asset price
as the spot price risk. The second-order effect is attributed to the future realized volatility.}, 
(ii) future realized volatility, (iii) future implied volatility
and (iv) sub-optimal withdrawal.

It is also worth noting the last two terms are missing in the classical attribution analysis
based on the BS greeks.
{\it Generally speaking, for any given model,
should it diverge from the market,
the classical greeks-based-attribution is incomplete, which is reflected by unexplained profit/loss.}
% (i), (ii) and (iv) can be monitored constantly as time progresses
% and any unexplained P\&L should be attributed to the future smile risk (iii);
% see Eqs. \eqref{eq:va_gamma_process}--\eqref{eq:va_suboptimal_process} of Appendix \ref{app:proof_multi_period}.
\end{remark}

\subsection{Separating Risk, Hedging, and Pricing Models}
\label{sec:separation}
Now we comment on the impacts of the model risk of the BS model from several aspects.
\begin{itemize}
\vspace{1ex}\item Firstly and foremost, Remark \ref{rem:breakeven} reveals that 
the BS model enables the insurer to speculate the future implied and realized volatility
via marking up/down the parameter $v$.
This distorts the incentives of the insurer and is undesirable from the regulator's perspective. 

\vspace{1ex}\item From the perspective of pricing, the problem with the BS model is more about its lack of flexibility.
There is only one degree of freedom (the parameter $v$) that can be utilized by the pricer to 
mark up/down the realized volatility and future smile risks {\it simultaneously}.
In other words, the insurer is not able to control the two risks {\it separately} 
despite that realized and implied volatility don't behave in line with each other in reality \cite{Bergomi2015}
and have different impacts on the insurer's P\&L (see Remark \ref{rem:pnl_leakage}).

\vspace{1ex}\item From the viewpoint of hedging, the hedger's perception of the P\&L is misled by the BS model. 
We recall from Remark \ref{rem:pnl_leakage}
that the presence of the future smile risk causes an instantaneous jump across the contract event date.
This is very annoying: 
it is likely that the hedger's positive cumulative P\&L is suddenly skewed up by leakage term \eqref{eq:pnl_leakage}.

\vspace{1ex}\item Remark \ref{rem:risk_attribution} 
reveals that the BS model can be used as a decent tool for risk attribution analysis.
Specifically, it precisely pinpoints all risk drivers behind the variable annuity;
see Remark \ref{rem:risk_attribution}.
\end{itemize}

\vspace{2ex}
\section{Numerical Examples}
\label{sec:num_studies}
This section conducts some numerical experiments to study the efficacy of the BS-delta-hedging
given the market does not respect the BS model assumptions.
\subsection{Contract Specification}
The specification of contract-related payoff functions is relegated to Appendix \ref{app:num_studies}
for the clarity of the presentation.
We confine our attention into the two-period case considered by Section \ref{sec:two_period_case} with $t_i=i, i=0,1,2$.
\subsection{Dynamics of the Market}
For illustration purposes, we postulate the market follows the Heston-type stochastic volatility model given by
\bas
\begin{cases}
\dd X(u) &= r X(u) \dd u + \sqrt{\alpha(u)} X(u)\dd W_1(u),\ u \in [t_0, t_1],\\
\dd \alpha(u) &= \kappa \[\theta - \alpha(u)\] \dd u + \nu \sqrt{\alpha(u)}\dd W_2(u),
\ \dd W_1(u) \dd W_2(u) = \rho \dd u,
\end{cases}    
\eas
with $u \in [t_0, t_1]$.
We adopt the Heston parameters given in the following table.
\begin{table}[ht]
\begin{center}
\caption{Heston parameters used for numerical experiments.}
\label{tab:heston_para}
% \vspace{.5ex}
% \renewcommand{\arraystretch}{1.2}
\begin{tabular*}{0.9\textwidth}{l@{\extracolsep{\fill}}llllll}
\toprule
Parameter &  $r$ & $\kappa$& $\nu$ & $\theta$& $\alpha(0)$& $\rho$ \\
\midrule
Value &  0.0 &0.1& 0.1& 0.04& 0.04& -0.69\\
\bottomrule
\end{tabular*}
\end{center}
\end{table}

In our numerical experiments,
we simulate the Heston process at equally-spaced time points 
$\{u_i\}_{i=0}^{N}$ with $u_0=t_0$, $u_{N}=t_1$ and step size $\Delta u = u_{i+1}-u_{i}$.
Further denote $X_{i}^{{\blue [m]}}$ as the value of $X(u_i)$ in $m$-th simulated path.
In the consequent numerical experiments, we simulate 100 scenarios in total,
which is sufficient for illustrating the main conclusion.

\subsection{Efficacy of the BS-Hedging}
We recall from Eqs. \eqref{eq:carry_pnl} and \eqref{eq:mtm_pnl} that
the cumulative P\&L of the BS-hedging strategy is given by
\bas
\textup{P\&L}(t) &=& \[\Pi(t) - V_{\bs}\(t, X(t), v\)\]\\
&&+ \[V_{\bs}\(t_1, X(t_1), v\) - V_{\bs}\(t_1, X(t_1), \xi_1\)\]{\blue\1\{t=t_1\}},\ \ t\in[t_0, t_1],
\eas
where 
\bas
\Pi(t) = \int_0^{t}\partial_x V_{\bs}(u, X(u), v) \dd X(u) + r\int_0^t\[\Pi(u) - \partial_x V_{\bs}(u, X(u), v)X(u)\]\dd u.
\eas

Then the cumulative P\&L at time $t_1$ along the $m$-th simulated path is approximately given by:
\bas
\tilde{\pnl}^{{\blue [m]}}(v) = \tilde{\Pi}^{{\blue [m]}}(u_N, v) - V_{\bs}\(u_N, X_{N}^{{\blue [m]}}, \xi_1^{{\blue [m]}}\)
\eas
where $\tilde{\Pi}$ is recursively defined by
\bas
\tilde{\Pi}^{{\blue [m]}}(u_{{\red i+1}}, v) = 
\Delta_{\bs}^{{\blue [m]}}(u_{\red i})\[\Delta X_i^{{\blue [m]}} - rX_i^{{\blue [m]}}\Delta u\]
+ \tilde{\Pi}^{{\blue [m]}}(u_{{\red i}}, v) \[1 + r\Delta u\],
\ \ 0\leq i \leq N-1,
\eas
with $\tilde{\Pi}^{{\blue [m]}}(u_{0}, v) = V_{\bs}\(0, X(t_0), v\)$ and
$
\Delta_{\bs}^{{\blue [m]}}(u_i) := \partial_x V_{\bs}\(u_i, X_i^{{\blue [m]}}, v\);
$
see Eq. \eqref{eq:bs_delta} of Appendix \ref{app:num_studies}
for the expression of $\partial_x V_{\bs}$.

Figure \ref{fig:density_plot_imp_vol} displays the histograms of the simulated values
of the BS implied variance of the contract $\xi_1^{{\blue [m]}}$
and the underlying asset price $X_N^{{\blue [m]}}$ at $t_1$ respectively.
We can see that the implied variance at a future time point is random
under the Heston model:
it can lie arbitrarily above or below any given BS variance parameter $v$
and thus introduces the future smile risk and P\&L leakage; see Remark \ref{rem:pnl_leakage}.
\begin{figure}[ht]
\begin{center}
\includegraphics[width=1.0\textwidth]{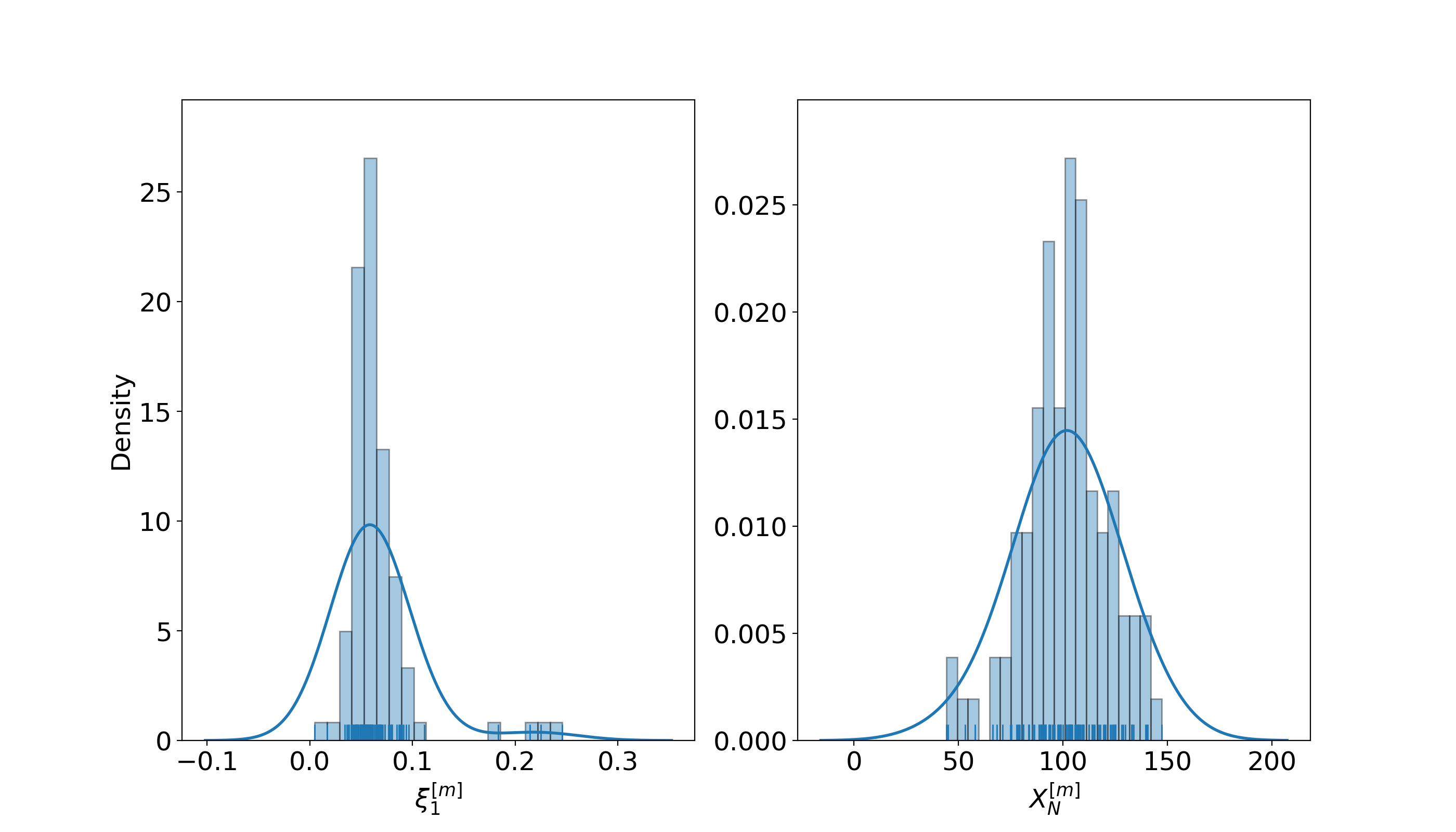}
\captionsetup{width=0.9\textwidth}
\caption{Histograms of implied variance of the contract and the underlying price at $t_1$.}
\label{fig:density_plot_imp_vol}
\end{center}
\end{figure}

Figure \ref{fig:density_plot_pnl} plots the histograms of the cumulative P\&L at $t_1$
with varying hedging frequency (daily vs monthly) and BS variance parameter 
$v$ ($0.04$ vs $0.09$).
We have two observations. Firstly, as one increases the hedging frequency 
from monthly to daily, the P\&L becomes more stable,
 as reflected by the more spiked shape of the density plot.
Such a variance reduction is due to the fact that the BS-Delta-hedging eliminates the spot price risk.
Secondly, as the hedger marks up the BS variance parameter $v$,
it is more likely to get positive cumulative P\&L.
This is in line with the conclusion of Proposition \ref{prop:carry_pnl}:
both P\&L leakage and slippage terms are monotone in $v$;
furthermore, they are positive should $v$ dominate the implied and realized variance over the course of the hedging.

\begin{figure}[ht]
\begin{center}
\includegraphics[width=1.0\textwidth]{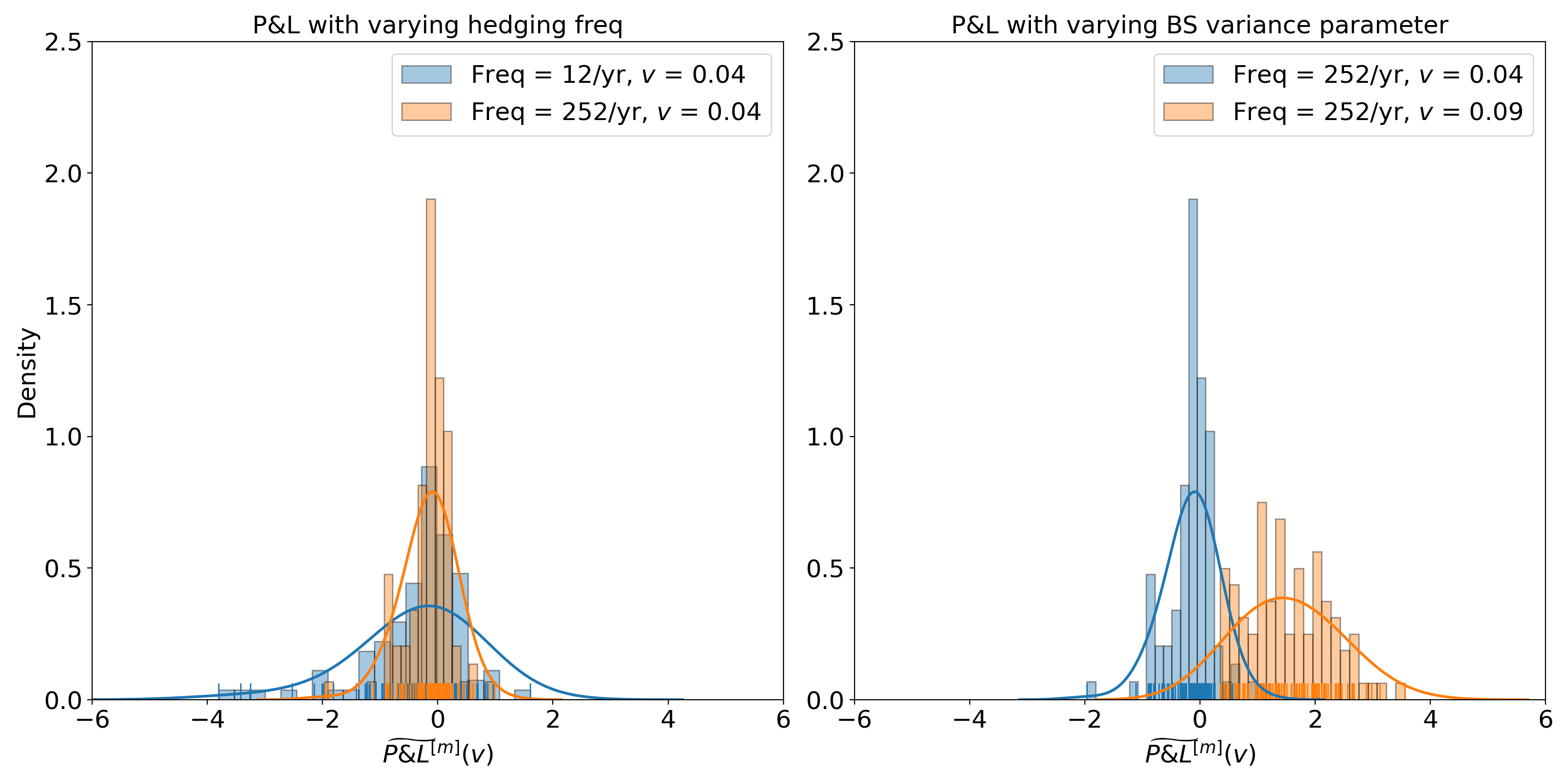}
\captionsetup{width=0.9\textwidth}
\caption{Histograms of the cumulative P\&L at $t_1$.}
\label{fig:density_plot_pnl}
\end{center}
\end{figure}

To sum up, despite that the market disagrees with the BS model assumptions,
the hedger can always stabilize her final P\&L by increasing hedging frequency
and even benefit from the mis-hedging (model risk) by speculatively marking up the BS variance parameter.

\vspace{2ex}
\section{Concluding Remarks}
\label{sec:conclusion}
We have shown that the fair value of the variable annuity can be decomposed into four parts
(i) BS model price,
(ii) valuation adjustment for future realized volatility,
(iii) valuation adjustment for future implied volatility smile,
(iv) and valuation adjustment for sub-optimal withdrawal risk.
This discloses that the risk of the variable annuity
can be exhaustively attributed to four corresponding factors.
The insurer can conservatively price (ii) and (iii) simultaneously 
by marking up the variance parameter
but has no control over (iv).

This paper also shows that the impact of model risk on the cumulative P\&L of 
the classical BS-based-hedging strategy
is reflected in two different ways: gradual slippage and instantaneous leakage.
There is even a chance that the hedger can benefit from taking the model risk.
Furthermore, the P\&L caused by spot price can always be eliminated by the strategy,
which is \textit{immune} to the model risk.

It is worth stressing that the primary thrust of this article is delineating the risk profile of variable annuities
rather than promoting or criticizing the BS model.
The BS model plays the role of an extractor for the spot price risk 
and our out-of-model adjustment formula further anatomizes the residual model risk.
Pinpointing all risk drivers paves the way to systematically access the advantages and limitations of more advanced models 
(local volatility model, stochastic volatility model, stochastic local volatility model, etc)
in pricing and risk-managing variable annuities.
Specifically, one may scrutinize how a given model prices in each risk segment
and accordingly, underprice/overprice the product as a whole.
This is left as a future research avenue.

Generally speaking, one may decompose the model-free price into any given model price
plus out-of-model adjustments.
It will be fruitful to explore different decomposition formulas on top of different models. 
However, the more fundamental question is: 
does such a decomposition decouples the risks with different natures
and allow the pricer to control them \textit{individually} by marking up/down some free parameter?
Our choice of the BS model as the extractor does not have no special significance.
A more fancy model does not necessarily give better pricing and risk-management of 
the variable annuity due to the less transparency of the model risk.
Putting it another way, the simpler the model is, the better grasp of the model risk is.

\vspace{2ex}
\section*{Acknowledgements}
The author is grateful to the inspiring discussions with Dr. Xi Tan and Professor Chengguo Weng.
%====================================================================================
\newpage

\setlength{\parskip}{0.5ex}
\hypertarget{LinkToAppendix}{\ }
\appendix
\vspace{-1.2cm}
\section{Appendix Companion to Section \ref{sec:num_studies}}
\label{app:num_studies}
\subsection{Contract Specification}
For the clarity of illustration, 
we consider a simple variable annuity contract specified as follows.
\begin{itemize}
\item[] \textit{Transition of investment account across withdrawal date}\quad
Across each withdrawal date, the investment account balance is reduced
by the withdrawal amount and accordingly,
\bas
K(x, a) = \max\[x-a,0\],\ a\in A_{n}(x)=\[0, x\].
\eas

\vspace{1ex}
\item[] \textit{Intermediate payoff function}\quad 
Given the withdrawal amount $a$ at $t_n$, the policyholder's reward is given by
\bas
f_n(x, a) = a(1-\eta),\ a \in A_{n}(x),
\eas
where $\eta \in [0,1]$ is withdrawal penalty.

\vspace{1ex}\item[] {\it Terminal payoff function}\quad 
At maturity, the policyholder can receive the balance of the investment account 
or the guaranteed withdrawal amount, whichever is larger.
Thus, the terminal payoff is given by
\bas
g(x) = x + \max\[G-x, 0\].
\eas
\end{itemize}

Throughout Section \ref{sec:num_studies}, we adopt the set of parameters in the table below.
\begin{table}[ht]
\begin{center}
\caption{Product parameters used for numerical experiments}
\label{tab:product_para}
\begin{tabular*}{0.5\textwidth}{c@{\extracolsep{\fill}}cc}
\toprule
Parameter &  $\eta$ & $G$ \\
\midrule
Value &  0.6 & 50\\
\bottomrule
\end{tabular*}
\end{center}
\end{table}
% \begin{table}[ht]
% \begin{center}
% \caption{Product parameters used for numerical experiments.}
% \label{tab:product_para}
% \vspace{.5ex}
% \renewcommand{\arraystretch}{1.2}
% \begin{tabular*}{0.5\textwidth}{l @{\extracolsep{\fill}} r}
% \toprule
% Parameter &   Value\\
% \midrule
% Withdrawal penalty $\eta$  &   $0.6$ \\
% Guaranteed payment $G$     &   $50$ \\
% \bottomrule
% \end{tabular*}
% \end{center}
% \end{table}

\subsection{BS Value Function and Delta}
It follows from Eq. \eqref{eq:Bellman_eq_BS} that $V_{\bs}(t_2, x, v) = \max[G, x]$ and
\bas
V_{\bs}(t_1, x, v) = \sup_{a \in A_1(x)} \[f_1(x,a) + K(x,a) + P_{\bs}(\delta, K(x,a), G, v)\],
\eas
where $P_{\bs}(\delta, x, k, v)$ denotes the BS put option price with spot price $x$, 
strike $k$, time-to-expiry $\delta$ and BS-variance parameter $v$.

The plot of the value function is displayed in Figure \ref{fig:bs_value_fun}
from which we can see how sensitive the function is to the BS variance parameter.
When the market does not follow the BS model, we recall from Proposition \ref{prop:carry_pnl} 
that the fair value of the contract is given
by the above BS value function with $v$ replaced by the prevailing implied variance at $t_1$
which can lies arbitrarily above or below $v$.

To avoid direct evaluation of the derivative of $V_{\bs}$, we adopt the likelihood ratio method
\cite{Broadie1996} to compute
\ba
\label{eq:bs_delta}
\partial_x V_{\bs}(u, x, v) = \e^{\bs}\[\frac{e^{-r(t_1 - u)}}{x\sqrt{v(t_1-u)}}d(u, x, v)V_{\bs}\(t_1, X(t_1), v\)\]
\ea
% \ba
% % \nonumber
% \label{eq:dollar_gamma}
% x^2\partial_{x}^2 V_{\bs}(u, x, v) = e^{-r(t_1 - u)}\e^{\bs}\[ V_{\bs}(t_1, X(t_1), v)
% \frac{d^2 - d\sqrt{v(t_1-u)}-1}{v (t_1 - u)}\]
% % &=& {\blue x} e^{-r(t_1 - u)}\e^{\bs}\[ V_{\bs}(t_1, {\blue X(t_1)/x}, v)
% % \frac{d^2 - d\sqrt{v(t_1-u)}-1}{v (t_1 - u)}\]
% \ea
where 
\bas
d(u, x, v) = \frac{\ln \[X(t_1)/x\] - \(r - v/2\)(t_1 -u)}{\sqrt{v(t_1-u)}}
\eas
and $\e^{\bs}$ denotes the expectation under which
% It is worth noting that the distribution of $X(t_1)/x$ has no dependency on $x$ because
% under the BS model
\bas
\ln \[X(t_1)/x\] \sim \mathcal{N}\((r-v/2)(t_1 - u), v(t_1 - u)\).
\eas

\begin{figure}[ht]
\begin{center}
\includegraphics[width=0.8\textwidth]{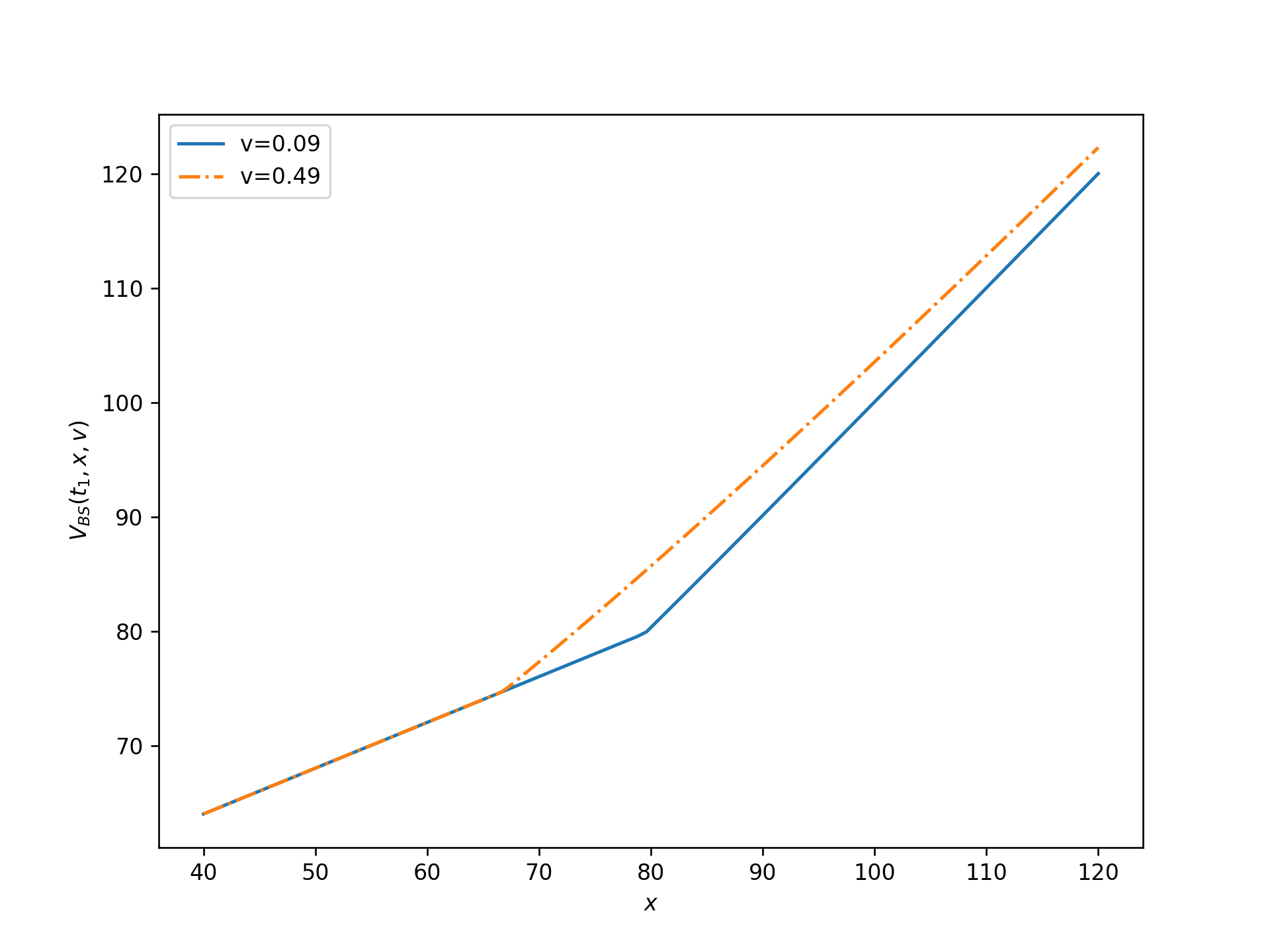}
\captionsetup{width=0.95\textwidth}
\caption{BS value function with varying variance parameters.}
\label{fig:bs_value_fun}
\end{center}
\end{figure}

\vspace{2ex}
\section{Technical Assumptions}
\label{app:assum}
Throughout this paper, we impose the following technical assumptions.
Consider a filtered probability space $\(\Omega, \(\F_{t}\)_{t\in[t_0, t_N]}, \q\)$ 
satisfying the usual conditions.
% \begin{assumption}
% The filtration $\mathbb{F}=(\F_n)_{t\in [0, t_N]}$ is generated by a controlled-Markov-process
% $\(Y(t)\)_{t\in[0, t_N]}$ defined recursively as follows.
% \end{assumption}

\begin{assumption}
\bas
\sup_{\mathbf{a}\in \mathcal{A}} \e^{\q}\[\sum_{n=1}^{N-1} |f_{n}(X_n, a_n)|\] < \infty
\ \textnormal{and}\ \sup_{\mathbf{a}\in \mathcal{A}} \e^{\q}\[g(X_{N})\],
\eas
where 
\bas
\mathcal{A}=\left\{\left.\mathbf{a}=\{a_{n}\}_{n=1}^{N-1}\right|
a_{n}\ \textnormal{is}\ \F_{t_n}\textnormal{-measurable and}\ a_n \in A_n(X_{n})
\right\}
\eas  
is the set of admissible controls.
\end{assumption}

\begin{assumption}
\label{assum:optimal_policy}
There exists $\mathbf{a}^{\star}= \{a_{n}^{\star}\}_{n=1}^{N-1}\in \mathcal{A}$ s.t. 
the supremum in Bellman equation \eqref{eq:Bellman_eq} is attained at $a_{n}^{\star}$ for $1\leq n \leq N-1$.
\end{assumption}

\begin{assumption}
The investment account of the variable annuity is tied to a tradable asset 
and accordingly, its value between two consecutive withdrawal dates is given by
\bas
X(t) = X_{n^{+}} \frac{S(t)}{S(t_n)},\ t\in{\blue (t_n}, t_{n+1}],
\ 0\leq n\leq N-1,
\eas
where $S(t)$ is the price process of the underlying asset.

Assume that $S(t)$ is a continuous semi-martingale
and does not pay dividends.
Then we have
\ba
\label{eq:martingale}
\e^{\q}\[\left.\dd X(u) - rX(u)\dd u\right|\F_t\] = 0,\ \ t_n < t\leq u\leq t_{n+1},
\ \ 0\leq n\leq N-1
\ea   
with $\q$ being the risk-neutral measure.
\end{assumption}
\begin{remark}
The zero-dividend assumption can be removed by replacing $r$ in Eq. \eqref{eq:martingale}
by $r-d$ with $d$ being the dividend yield.
\end{remark}

\begin{assumption}
The value function under the BS model $V_{\bs}(t_n, x)$ defined via Eq. \eqref{eq:Bellman_eq_BS} 
is convex in $x$. Specially, when $n=N$, the terminal payoff $V_{\bs}(t_n, x) = g(x)$ is a convex function.
\end{assumption}
\begin{remark}
The above assumption ensures the BS-implied variance in the sequel Definition 
\ref{def:imp_var} is well-defined.
The BS value function of the variable annuity is convex under very mild conditions;
see e.g. \cite{Azimzadeh2015,Huang2016,Huang2017,Shen2020}.
The convexity is a sufficient condition for the existence of the implied variance 
and thus can be further relaxed.
\end{remark}

\vspace{2ex}
\section{Proof of Theorem \ref{thm:risk_decompose}}
\label{app:proof_risk_decompose}
\begin{proof}
Recall that when $N=2$ we have $V(t_2) = g(X(t_2)).$
Then it follows from Eq. \eqref{eq:post_withdrawal_value} that  
\bas
C_{1}(x) = \e_{1,x}^{\q} \[e^{-r\delta}V(t_2)\] 
= \e_{1,x}^{\q} \[e^{-r\delta}g(X(t_2))\] = C_{1}^{\bs}\(x,\xi_{1}\),
\eas
where the last equality is by the definition of $\xi_1$.
By Eqs. \eqref{eq:Bellman_eq} and \eqref{eq:Bellman_eq_BS}, we get
\bas
V(t_1) &=& \sup_{a \in A_1(X_1)} \[f_1(X_1, a) + C_1\(K(X_1, a)\)\] \\
&=& \sup_{a \in A_1(X_1)} \[f_1(X_1, a) + C_1^{\bs}\(K(X_1, a), {\blue \xi_1}\)\] = V_{\bs}\(t_1, X_1, {\blue\xi_1}\).
\eas
Then it follows from Eq. \eqref{eq:martingale_eq} that
\bas
V(t) = \e^{\q}\[\left. e^{-r (t_1 -t)}V(t_1)\right|\F_t\]
= \e^{\q} \[\left. e^{-r (t_1 -t)}V_{\bs}\(t_1, X(t_1), {\blue v}\)\right|\F_t\] + V_{2}(t),
\eas
where the last equality follows by Eq. \eqref{eq:smile_va}.

On the other hand, the BS equation \eqref{eq:Bellman_eq} in conjunction with the It\^{o}'s lemma implies
\bas
e^{-r (t_1 -t)}V_{\bs}\(t_1, X(t_1), v\) - V_{\bs}(t, X(t), v)
&=& \int_{t}^{t_1}e^{-r(u - t)} \frac{1}{2}\partial_{x}^2V_{\bs}\[\dd\[X\](u)-vX^2(u)\dd u\]\\
&&+ \int_{t}^{t_1} e^{-r(u-t)}\partial_x V_{\bs} \[\dd X(u) - r X(u)\dd u\].
\eas
Taking expectations on both sides of the above equation gives
\bas
\e^{\q} \[\left. e^{-r (t_1 -t)}V_{\bs}\(t_1, X(t_1), {\blue v}\)\right|\F_t\]
&=&  V_{\bs}(t, X(t), {\blue v}) + V_1(t)\\ 
&&+ \e^{\q}\[\left.\int_{t}^{t_1} e^{-r(u-t)}\partial_x V_{\bs} \[\dd X(u) - r X(u)\dd u\]\right|\F_t\]\\
&=& V_{\bs}(t, X(t), v) + V_1(t)
\eas
where the last equality follows by Eq. \eqref{eq:martingale}.

Putting the last three displays together yields Eq. \eqref{eq:va_eq}. This completes the proof.
\end{proof}

\section{Proof for Proposition \ref{prop:carry_pnl}}
\label{app:proof_carry_pnl}
\begin{proof}
Let $\Pi(u)$ be the value of the insurer's hedges as of time $u$.
In accordance with the setup of the hedging strategy {\bf (H1)}--{\bf (H3)}, the hedges are made up of the underlying asset $X(u)$ and cash position $B(u)$:
\bas
\Pi(u) = \Delta_{\bs}(u)X(u) + B(u),
\ \ \textup{with}\ \ \Delta_{\bs}(u) := \partial_x V_{\bs}(u, X(u), v)
\eas
which satisfies the self-financing condition:
\ba
\label{eq:self-financing}
\dd \Pi(u) = \Delta_{\bs}(u) \dd X(u) + rB(u) \dd u,
\ea
and the initial cost constraint $\Pi(0)=V_{\bs}(0, X_0, v)$.

In the sequel, we prove Eqs. \eqref{eq:carry_pnl} and \eqref{eq:mtm_pnl} consequently.
The first equality of Eq. \eqref{eq:carry_pnl} follows by {\bf (H2)}.
By It\^{o}'s lemma, we get
\bas
\dd V_{\bs}(u, X(u), v) = \partial_t V_{\bs} \dd u + \Delta_{\bs}(u) \dd X(u) + 
\frac{1}{2}\partial_{x}^2V_{\bs}\[\dd X(u)\]^2,
\eas
which in conjunction with Eq. \eqref{eq:self-financing} yields
\bas
\dd \[V_{\bs}(u, X(u),v)- \Pi(u)\] &=& \partial_t V_{\bs} \dd u + 
\frac{1}{2}\partial_{x}^2V_{\bs}\[\dd X(u)\]^2\\ 
&&- r\[\Pi(u) - \Delta_{\bs}(u) X(u)\]\dd u.
\eas
Recall that $V_{\bs}$ satisfies Eq. \eqref{eq:BS_eq}. Plugging it into the above equation yields
\bas
\dd \[V_{\bs}(u, X(u), v) - \Pi(u)\] &=&
\frac{1}{2}\partial_{x}^2V_{\bs}\[\dd\[X\](u)-vX^2(u)\dd u\]\\
&&+ r\[V_{\bs}(u, X(u), v) - \Pi(u)\]\dd u.
\eas
Integrating the above equation over $[t_0, t]$ and exploiting the fact that $e^{-r u} X(u)$ is a martingale implies 
\bas
V_{\bs}(t, X(t), v) - \Pi(t) &=& e^{r(t-t_0)}\[V_{\bs}(0, X_0, v) - \Pi(0)\] \\
&&+ \frac{1}{2}e^{r(t-t_0)} \int_{t_0}^{t} e^{-r(u-t_0)}\partial_{x}^2V_{\bs}\[\dd\[X\](u)-vX^2(u)\dd u\].
\eas
Then Eq. \eqref{eq:carry_pnl} follows by the fact that $V_{\bs}(0, X_0, v) = \Pi(0)$.

Finally, Eq. \eqref{eq:mtm_pnl} follows by {\bf (H3)}.
This completes the proof of Proposition \ref{prop:carry_pnl}.
\end{proof}

\vspace{2ex}
\section{Proof of Theorem \ref{thm:va_multi_period}}
\label{app:proof_multi_period}
\subsection{Preliminaries}
% Throughout the appendix, we use the shorthand notations $V_{\bs}(t,x):=V_{\bs}(t,x,v)$ and $C_n^{\bs}(x):=C_n^{\bs}(x,v)$
% to denote the BS value and post-withdrawal functions
% with the variance parameter {\blue fixed as $v$} over the entire contract life $[t_0, t_N]$.
\begin{definition}
\label{def:quasi_BS_post_withdrawal}
Consider the following BS-type PDE:
\ba
\label{eq:BS_eq_xi}
\begin{cases}
u|_{t=t_{n+1}}=V_{\bs}(t_{n+1}, x, {\spicy v}),&\\
\partial_t u + rx \partial_x u + {\blue \frac{\xi}{2}}x^2\partial^2_{xx} u - r u = 0,\ \ t\in[t_n, t_{n+1}),\\
\end{cases}
\ea
We define the function 
\ba
\tilde{C}_n(x, {\blue \xi}) = u\(t_{n}^{+}, x\),\ \ 
1\leq n \leq N-1.
\ea
\end{definition}
\begin{remark}
\label{rem:20220425}
It is worth stressing the difference between the functions $\tilde{C}_n(\cdot, \xi)$ and $C_{n}^{\bs}(\cdot)$ (see Eq. \eqref{eq:post_withdrawal_BS}).
\begin{enumerate}
\item They are both defined via the BS-type PDE but with BS variance parameters $\xi$ and $v$, respectively; see Eqs. \eqref{eq:BS_eq} and \eqref{eq:BS_eq_xi} respectively.

\item The time boundary conditions at $t=t_{n+1}$ of the two PDEs are the same.
\end{enumerate}
See also the following Lemma \ref{lemma:20220423}.
\end{remark}

\begin{definition}
\label{def:imp_var}
We say $\xi_n$ is the BS-implied variance of the European option with payoff $V_{\bs}\(t_{n+1}, X_{n+1}, v\)$
observed at $t_n$ if it satisfies the following equation:
\ba
\label{eq:imp_var}
\e_{n,x}^{\q}\[e^{-r\delta}V_{\bs}\(t_{n+1}, X_{n+1}, v\)\] = \tilde{C}_n(x, \xi_n),
\ea
where $\tilde{C}_n(x, \xi)$ is defined in Definition \ref{def:quasi_BS_post_withdrawal}
and we recall that $\e_{n,x}^{\q}[\cdot]:=\e^{\q}\[\cdot|\mathcal{G}_{t}^{x}\]$
with 
$
\mathcal{G}_{n}^{x}:=\sigma\(\{X_{n^{+}}=x\}\bigcup \mathcal{F}_{t_n}\).
$
\end{definition}

Let 
\ba
\label{eq:obj_BS}
J_n(x, a, \theta) = f_n(x,a) + \tilde{C}_n\(K(x,a), \theta\).
\ea
and 
\ba
\label{eq:quasi_BS_value_fun}
\tilde{V}_n(x,\theta) = \sup_{a \in A_n(x)} J_n(x, a, \theta).
\ea
% Then it follows from Eq. \eqref{eq:Bellman_eq_BS} that $V_n^{\bs}(x,) = \sup_{a \in A_n(x)} J_n(x, a, v)$.
% \bas
% \tilde{a}_n = \arg \sup_{a \in A_n(x)} J_n(x, a, v).
% \eas
\begin{lemma}
\label{lemma:20220423}
$\tilde{C}_{n}(x, v) = C_{n}^{\bs}(x, v)$ and $\tilde{V}_n(x) = V_{\bs}(t_n, x, v)$ for $1\leq n \leq N-1$
where $C_n^{\bs}(x,v)$ and $V_{\bs}(t_n,x,v)$ are given by Eqs. \eqref{eq:post_withdrawal_BS} and \eqref{eq:Bellman_eq_BS} respectively.
\end{lemma}
\begin{proof}
The proof directly follows from Remark \ref{rem:20220425} and thus is omitted.
\end{proof}
Next, we define the valuation adjustments (VAs) for future smile and sub-optimal withdrawal risks respectively.
\begin{definition}
\label{def:va_adj}
\begin{enumerate}
\vspace{1ex} \item The valuation adjustment for the future realized volatility is defined as 
\ba
\label{eq:va_gamma_process}
V_1(t) = \e^{\q}\[\left.\int_{t}^{t_n} e^{-r(u-t)}\frac{1}{2}\partial_{x}^2 V_{\bs} \[\dd [X](u) - X^2(u)\dd u\]\right|\F_t\],
\ \ t\in(t_{n-1}, t_{n}],
\ea
with the convention $V_1(t_n) =0$ for $0\leq n \leq N-1$.
\vspace{1ex}\item The valuation adjustment for the future smile risk is defined as
\ba
\label{eq:va_future_smile_process}
V_2(t) = \e^{\q}\[\left.\sum_{n=1}^{N-1}e^{-r(t_n - t)}
\[\tilde{V}_n(X_n, {\blue \xi_n}) - V_{\bs}(t_n, X_n, v) \]\1_{\{t \leq t_n\}} \right|\F_t\].
\ea
\vspace{1ex}\item
Let $a_n^{\star}$ be the optimal policy at $t_n$ of \eqref{eq:Bellman_eq}; 
see Assumption \ref{assum:optimal_policy}. 
Suppose $\tilde{a}_n$ satisfies $\tilde{V}_n(x,\xi_n) = J_n\(X_n, \tilde{a}_n, \xi_n\)$.
The valuation adjustment for suboptimal withdrawal risk is defined as
\ba
\label{eq:va_suboptimal_process}
V_3(t) &=& \e^{\q} \[\left.\sum_{n=1}^{N-1}e^{-r(t_n - t)}
\[ J_n\(X_n, {\blue a_n^{\star}}, \xi_n\) - J_{n}\(X_n, {\blue \tilde{a}_n}, \xi_n\)\]\1_{\{t \leq t_n\}} \right|\F_t\].
\ea
\end{enumerate}
\end{definition}

\begin{lemma}
The valuation adjustment for the future smile risk satisfies the following recursion system:
\ba
\label{eq:va_future_smile}
V_2(t_n) = 
\begin{cases}
0, & n=N,\\
\tilde{V}_n(X_n, {\blue \xi_n}) - V_{\bs}(t_n, X_n, {\blue v}) + 
C_{2,n}\(K(X_n, a_n^{\star})\),&1\leq n \leq N-1,
\end{cases}
\ea
where
\ba
\label{eq:exp_future_smile}
C_{2, n}(x) = \e_{n,x}^{\q}[e^{-r\delta}V_2(t_{n+1})].
\ea
Furthermore,
\ba
\label{eq:va_suboptimal_withdrawal}
V_3(t_n) = 
\begin{cases}
0, & n=N,\\
J_n\(X_n, {\blue a_n^{\star}}, \xi_n\) - J_{n}\(X_n, {\blue \tilde{a}_n}, \xi_n\)
+C_{3,n}\(K(X_n, a_n^{\star})\),&1\leq n \leq N-1,
\end{cases}
\ea
where
\ba
\label{eq:exp_suboptimal_withdrawal}
C_{3, n}(x) = \e_{n,x}^{\q}[e^{-r\delta}V_3(t_{n+1})].
\ea
\end{lemma}
\begin{proof}
The conclusion follows from a straightforward application of the Tower property of the conditional expectation.
\end{proof}

\begin{lemma}
\label{lemma:va_eq}
\ba
\label{eq:20220423}
V(t_n) = V_{\bs}(t_n, X_n, v) + V_2(t_n) + V_3(t_n),
\ \ 1 \leq n \leq N-1,
\ea
where $V_2(t_n)$ and $V_3(t_n)$ are given by Eqs. \eqref{eq:va_future_smile_process} and \eqref{eq:va_suboptimal_process} respectively.
\end{lemma}
\begin{proof}
We prove the lemma via an induction argument. 
\begin{enumerate}
\vspace{1ex}    \item
\begin{enumerate}
\item By Definition \ref{def:imp_var} and Eq. \eqref{eq:post_withdrawal_value}, we have $C_{N-1}(x) = \tilde{C}_{N-1}(x, \xi_{N-1})$.
Plugging this into Eq. \eqref{eq:Bellman_eq} implies that
\bas
V(t_{N-1}) &=& \sup_{a\in A_{N-1}(X_{N-1})}\left[f_{N-1}\(X_{N-1}, a\) + C_{N-1}\(K\(X_{N-1}, a\)\)\right]\\
&=& \sup_{a\in A_{N-1}(X_{N-1})}\[f_{N-1}\(X_{N-1}, a\) + \tilde{C}_{N-1}\(K\(X_{N-1}, a\), {\blue\xi_{N-1}}\)\]\\
&=& \sup_{a\in A_{N-1}(X_{N-1})}J_{N-1}\(X_{N-1}, a, \xi_{N-1}\)\\
&=& \tilde{V}_{N-1}(X_{N-1}, {\blue\xi_{N-1}}),
\eas
where the last equality is by Eq. \eqref{eq:quasi_BS_value_fun}
and thus $a_{N-1}^{\star} = \tilde{a}_{N-1}$. 

\item In view of Eqs. \eqref{eq:va_future_smile} and \eqref{eq:exp_future_smile}, we get
\bas
V_2(t_{N-1}) = \tilde{V}_{N-1}(X_{N-1}, {\blue\xi_{N-1}}) - V_{\bs}(t_{N-1}, X_{N-1}, {\blue v}).
\eas

\item Furthermore, combining Eqs. \eqref{eq:va_suboptimal_withdrawal} and \eqref{eq:exp_suboptimal_withdrawal} implies
that $V_3(t_{N-1})=0$.
\end{enumerate}
Putting (a)-(c) together implies Eq. \eqref{eq:20220423} holds for $n=N-1$.

\vspace{1ex}\item
Now by induction we assume Eq. \eqref{eq:20220423} holds for $n$.
% W.l.o.g. we assume $r=0$ in the sequel.
Plugging Eq. \eqref{eq:20220423} into Eq. \eqref{eq:post_withdrawal_value} gives
\bas
C_{n-1}(x) &=&\e_{n-1,x}^{\q} \[e^{-r\delta} V_{\bs}(t_{n}, X_n, v)\]
+ e^{-r\delta}\e_{n-1,x}^{\q} \[ V_2(t_n) + V_3(t_n)\]\\
% &=& \tilde{C}_{n-1}\(K(x,a), {\blue \xi_{n-1}}\) + \e_{n-1,x}^{\q} \[ V_2(t_n) + V_3(t_n)\]\\
&=& \tilde{C}_{n-1}\(K(x,a), {\blue \xi_{n-1}}\) + C_{2,n-1}\(K(x,a)\) + C_{3,n-1}\(K(x, a)\)
\eas
where the last equality follows by 
Eqs. \eqref{eq:imp_var}, \eqref{eq:exp_future_smile} and \eqref{eq:exp_suboptimal_withdrawal}.
% and the second equality holds by the Tower property of the conditional expectation.

The above display in conjunction with Eq. \eqref{eq:Bellman_eq} gives
\bas
V(t_{n-1}) &=& f_{n-1}\(X_{n-1}, a_{n-1}^{\star}\) + \tilde{C}_{n-1}\(K\(X_{n-1},a_{n-1}^{\star}\), \xi_{n-1}\) \\
&&+\sum_{j=2}^3 C_{j,n}\(K\(X_{n-1},a_{n-1}^{\star}\)\)\\
&=& J_n\(X_{n-1}, {\blue a_{n-1}^{\star}}, \xi_{n-1}\) +\sum_{j=2}^3 C_{j,n}\(K\(X_{n-1},a_{n-1}^{\star}\)\)\\
&=& \tilde{V}_{n-1}\(X_{n-1}, {\blue v}\) \\
&& + \underbrace{\tilde{V}_{n-1}\( X_{n-1}, {\blue \xi_{n-1}}\) -  V_{\bs}\(t_{n-1}, X_{n-1}, {\blue v}\) + C_{2,n}\(K\(X_{n-1},a_{n-1}^{\star}\)\)}_{=V_2(t_{n-1})}\\\
&&+ \underbrace{J_n\(X_{n-1}, {\blue a_{n-1}^{\star}}, \xi_{n-1}\) - J_{n}\(X_{n-1}, {\blue \tilde{a}_{n-1}}, \xi_{n-1}\)
+ C_{3,n}\(K\(X_{n-1},a_{n-1}^{\star}\)\)}_{=V_3(t_{n-1})}\\
&=& V_{\bs}\(t_{n-1}, X_{n-1}\) + V_2(t_{n-1}) + V_3(t_{n-1}),
% &=& V_{\bs}\(t_{n-1}, X_{n-1}, {\blue v}\) \\
% &&+ \tilde{V}_{n-1}\( X_{n-1}, {\blue \xi_{n-1}}\) -  V_{\bs}\(t_{n-1}, X_{n-1}, {\blue v}\) 
% + C_{2,n}\(K\(X_{n-1},a_{n-1}^{\star}\)\)\\
% && + V_3(t_{n-1}),
\eas
where the second equality is by Eq. \eqref{eq:obj_BS} 
and the last equality follows by Eqs. \eqref{eq:va_future_smile} and \eqref{eq:va_suboptimal_withdrawal} 
in together with Lemma \ref{lemma:20220423}.
This proves Eq. \eqref{eq:20220423} for $n-1$. 
\end{enumerate}
\end{proof}
\subsection{Proof of Theorem \ref{thm:va_multi_period}}
\begin{proof}
Recall from Lemma \ref{lemma:20220423} that
\bas
V(t_n) = V_{\bs}(t_n, X_n, v)  + V_2(t_n) + V_3(t_n),
\ \ 1 \leq n \leq N-1.
\eas
Plugging the above display into Eq. \eqref{eq:martingale_eq} gives
\ba
\nonumber
V(t) &=& \e^{\q}\[\left. e^{-r(t_n - t)}V_{\bs}(t_n, X_n, v)\right| \F_t\]
+ \sum_{j=2}^3\e^{\q}\[\left. e^{-r(t_n - t)}V_j(t_n)\right| \F_t\]\\
\label{eq:20220508}
&=& \e^{\q}\[\left. e^{-r(t_n - t)}V_{\bs}(t_n, X_n, v)\right| \F_t\] 
+ \sum_{j=2}^3 V_j(t),
\ea
where the last equality follows by the Tower property; see Eqs. \eqref{eq:va_future_smile_process}
and \eqref{eq:va_suboptimal_process}.

On the other hand, applying It\^{o}'s lemma to $e^{-r(t_n - t)}V_{\bs}(t_n, X_n, v)$ gives
\bas
e^{-r(t_n - t)}V_{\bs}(t_n, X_n, v) - V_{\bs}(t, X(t), v) 
&=& \int_{t}^{t_n} e^{-r(u-t)}\[\partial_t V_{\bs}-rV_{\bs}\]\dd u\\
&&+\int_{t}^{t_n} e^{-r(u-t)}\Delta_{\bs}(u) \dd X(u)\\
&&+ \int_{t}^{t_n} e^{-r(u-t)}\frac{1}{2}\partial_{x}^2V_{\bs} \dd [X](u),
\eas
for $u \in [t, t_n]$.
Recall that $V_{\bs}$ satisfies the BS equation:
\bas
\partial_t V_{\bs} + rx \partial_x V_{\bs} + \frac{v}{2}x^2\partial^2_{xx} V_{\bs}- rV_{\bs} = 0.
\eas
Hence,
\bas
e^{-r(t_n - t)}V_{\bs}(t_n, X_n, v) - V_{\bs}(t, X(t), v) 
&=& \int_{t}^{t_n} e^{-r(u-t)}\partial_x V_{\bs} \[\dd X(u) - r X(u) \dd u\]\\
&&+ \int_{t}^{t_n} e^{-r(u-t)}\frac{1}{2}\partial_{x}^2 V_{\bs} \[\dd [X](u) - vX^2(u)\dd u\].
\eas
Taking expecations on both sides of the above equation gives
\bas
\e^{\q}\[\left. e^{-r(t_n - t)}V_{\bs}(t_n, X_n, v)\right| \F_t\] 
&=& V_{\bs}(t, X(t), v) + V_{1}(t) \\
&&+ \e^{\q}\[\left.\int_{t}^{t_n} e^{-r(u-t)}\partial_x V_{\bs} \[\dd X(u) - r X(u) \dd u\]\right|\F_t\]\\
&=& V_{\bs}(t, X(t), v) + V_{1}(t),
\eas
where the last equality follows by Eq. \eqref{eq:martingale}.

The above display in conjunction with Eq. \eqref{eq:20220508} proves Theorem \ref{thm:va_multi_period}.
\end{proof}
% %====================================================================================
% %====================================================================================
% %====================================================================================

\vspace{0.8cm}
\newpage

\bibliographystyle{plain}
\bibliography{ref}

\vspace{0.4cm}

\end{document}